\documentclass{amsart}
\usepackage{times}
\usepackage{amssymb,amsmath,amscd, amsthm,stmaryrd,verbatim, mathrsfs}

\usepackage{color}

\newtheorem{Thm}{Theorem}

\numberwithin{equation}{section}

\newcommand{\mb}{\mathbf}
\newcommand{\bb}{\mathbb}
\newcommand{\ms}{\mathscr}
\newcommand{\mr}{\mathrm}

\usepackage{hyperref}
\begin{document}

\title{Lorentz Group and Oriented MICZ-Kepler Orbits}
\author{Guowu Meng}

\address{Department of Mathematics, Hong Kong Univ. of Sci. and
Tech., Clear Water Bay, Kowloon, Hong Kong}
%\curraddr{Institute for Advanced Study, Einstein Drive, Princeton, New Jersey 08540 USA}

\email{mameng@ust.hk}
\thanks{The author was supported by the Hong Hong Research Grants Council under RGC Project No. 603110 and the Hong Kong University of Science and Technology under DAG S09/10.SC02.}

%    General info
%\subjclass[2000]{Primary 22E46, 22E70; Secondary 81S99, 51P05}

\date{November 15, 2011}


\maketitle
\begin{abstract}
 The MICZ-Kepler orbits are the non-colliding orbits of the MICZ Kepler problems (the magnetized versions of the Kepler problem). The oriented MICZ-Kepler orbits can be parametrized by the canonical angular momentum $\mb L$ and the Lenz vector $\mb A$, with the parameter space consisting of the pairs of 3D vectors $(\mb A, \mb L)$ with ${\mb L}\cdot {\mb L} > (\mb L\cdot \mb A)^2$.  The recent 4D perspective of the Kepler problem yields a new parametrization, with the parameter space consisting of the pairs of Minkowski vectors $(a,l)$ with $l\cdot l =-1$, $a\cdot l =0$,  $a_0>0$. Here, $a_0$ is the temporal component of $a$. 
 
This new parametrization of orbits implies that the MICZ-Kepler orbits of different magnetic charges are related to each other by symmetries: \emph{${\mr {SO}}^+(1,3)\times {\bb R}_+$ acts transitively on both the set of oriented elliptic MICZ-Kepler orbits and the set of oriented parabolic MICZ-Kepler orbits}. This action extends to ${\mr {O}}^+(1,3)\times {\bb R}_+$, the \emph{structure group} for the rank-two Euclidean Jordan algebra whose
 underlying Lorentz space is the Minkowski space.   
 \end{abstract}

%\tableofcontents

\section {Introduction}
The Kepler problem is a two-body dynamic problem with an attractive force obeying the inverse square law. Mathematically it can be reduced to the one-body dynamic problem with the Lagrangian
\begin{eqnarray}
\ms L= {1\over 2}{\mb r}'^2+{1\over r}\nonumber
\end{eqnarray} 
where ${\mb r}$ is a function of $t$ taking value in ${\bb R}^3_*:={\bb R}^3-\{{\mb 0}\}$, ${\mb r}'$ is the velocity vector and $r$ is the length of $\mb r$. We shall also refer to this later one-body dynamic problem as the Kepler problem.

\smallskip
The (classical) MICZ Kepler problem with magnetic charge $\mu\in {\bb R}$, a natural mathematical generalization of the Kepler problem,  is the one-body dynamic problem with the Lagrangian \cite{MICZ}
\begin{eqnarray}
\ms L= {1\over 2}{\mb r}'^2+{1\over r} -{\ms A}\cdot {\mb r}'-{\mu^2\over 2r^2}\nonumber
\end{eqnarray} 
where ${\ms A}$ is the magnetic potential such that ${\mb B}:=\nabla\times {\ms A}=\mu{{\mb r}\over r^3}$.  Then the equation of motion is
\begin{eqnarray}\label{EM}
{\mb r}'' = - {\mb r}'\times {\mb B}+\left({\mu^2\over r^4}-{1\over r^3}\right){\mb r}. 
\end{eqnarray}

For each MICZ Kepler problem, our main concern here is its \emph{orbits}. Here, an orbit  is just the trace of a solution to equation (\ref{EM}). Actually we are only interested in the \emph{MICZ-Kepler orbits}, i.e., the non-colliding orbits. More precisely, we are only interested in the oriented MICZ-Kepler orbits --- the non-colliding solutions to equation (\ref{EM}) module the time translation.

\smallskip
The main purpose of this note is to add the following new piece of knowledge about the MICZ-Kepler orbits:  \emph{${\mr {SO}}^+(1,3)\times {\bb R}_+$ acts transitively on both the set of oriented elliptic MICZ-Kepler orbits and the set of oriented parabolic MICZ-Kepler orbits}. Here ${\mr {SO}}^+(1,3)$ is the identity component of ${\mr {SO}}(1,3)$ and ${\bb R}_+$ be the multiplicative group of positive real numbers. This action actually extends to ${\mr {O}}^+(1,3)\times {\bb R}_+$. Here, $\mr O^+(1,3)$ is the maximal subgroup of $\mr O(1,3)$ which leaves the future light cone invariant. The significance of this new piece of knowledge is that, elliptic (or parabolic) MICZ Kepler orbits of various magnetic charges are related by symmetries. 

\vskip 10pt
In section \ref{1stP}, we present the old parametrization for the oriented MICZ-Kepler orbits. In the process, a new formulation for the oriented MICZ-Kepler orbits is introduced. This new formulation, originated from the recent Jordan algebra approach \cite{meng09} to the Kepler problem, naturally leads to a new parametrization for the oriented MICZ-Kepler orbits. The aforementioned new piece of knowledge about the MICZ-Kepler orbits, stated as Theorem \ref{main} in section  \ref{2ndP},  becomes self-evident from this new parametrization.  

\vskip 10pt
Before we get into the details in the next two sections, let us fix our notations and conventions for the 3D space and the Minkowski space.

By the 3D space we mean ${\bb R}^3$ together with the usual dot product. By convention, an element in this space  is written by a Latin letter in bold. For ${\mb a}\in {\bb R}^3$, we
write ${\mb a} = (a_1, a_2, a_3)$ if we wish to write out its components explicitly. In the explicit form, the dot product of $\mb a $ with $\mb b$ is
$$
{\mb a}\cdot {\mb b}=a_1b_1+a_2b_2+a_3b_3.
$$
As usual,  we shall write the length of $\mb a$ as $a$. With this in mind, the old parametrization space for the oriented MICZ-Kepler orbits is
$$
{\ms M}=\{(\mb A, \mb L)\in {\bb R}^3\times {\bb R}^3\mid L^2> (\mb L\cdot \mb A)^2\}.
$$
The Minkowski space is $({\bb R}^4, \cdot)$ where $\cdot$ is the Lorentz dot product.    By convention, an element in this space is written by a small Latin letter, and we
write $a = (a_0, a_1, a_2, a_3)$ if we wish to write out the components of $a$ explicitly. We also write $a=(a_0, \mb a)$.  In the explicit form, the Lorentz dot product of $a$ with $b$ is
$$
a\cdot b=a_0b_0-{\mb a}\cdot {\mb b}.
$$
As usual, we shall write $a\cdot a$ as $a^2$ for simplicity. Here is a word of warning: $a$ either denote a $4$-vector or the length of the $3$-vector $\mb a$. One should have no problem to determine which is which according to the context. 
 With this in mind, the new parametrization space for the oriented MICZ-Kepler orbits is
$$
{\mathcal M}=\{(a, l)\in {\bb R}^4\times {\bb R}^4\mid a\cdot l =0, l^2 = -1, a_0 >0\}.
$$

The two parametrization spaces are smoothly equivalent:
\begin{eqnarray}
{\mathcal M}& \to & {\ms M}\cr
(a, l) &\mapsto & ({\mb a\over a_0}, {\mb l\over \sqrt{a_0}}).\nonumber
\end{eqnarray}
While the action of ${\mr {O}}^+(1,3)\times {\bb R}_+$ on $(a, l)$ is linear and apparent:
$$(\Lambda, \lambda)\cdot (a, l) = (\lambda\cdot \Lambda a, \Lambda l) ,$$
the corresponding  action on $({\mb A}, {\mb L})$ is nonlinear, that is probably why this action has not been discovered up until now. Since the positive sign of $a_0$ is kept under the group action if and only if either $a^2 >0$ (elliptic orbits) or $a^2 =0$ (parabolic orbits),  there is no similar action on the oriented hyperbolic orbits $(a^2 < 0$).

\section{The old parametrization of orbits}\label{1stP}
The orbits are very easy to find. To start,  we observe that
\begin{eqnarray}
{\mb L} = {\mb r}\times {\mb r}'+\mu{{\mb r}\over r}, \quad {\mb A} = {\mb L}\times {\mb r}'+{{\mb r}\over r}
\end{eqnarray}
are constants of motion and 
\begin{eqnarray}
{\mb L}\cdot {\mb A} = \mu.
\end{eqnarray}
Here, $\mb L$ is called the (canonical) angular momentum and $\mb A$ is called the Lenz vector. To get an orbit, all we need to do is to take the dot product of $\mb r$ with $\mb L$ and $\mb A$ to produce equations
\begin{eqnarray}\label{orbit}
r-{\mb A}\cdot {\mb r} = L^2-\mu^2, \quad {\mb L}\cdot {\mb r} = \mu r.
\end{eqnarray}
Here $L$ is the length of $\mb L$. Note that, $L^2-\mu^2= |{\mb r}\times {\mb r}'|^2\ge 0$ and the equality holds if and only if the orbit is a colliding orbit. Since we are only interested in the non-colliding orbits, from now on we always assume that $L^2> \mu^2$. In this case, since $\mb r'$ vanishes at nowhere, $\mb r(t)$ is defined for $t\in \bb R$ and the orbit must be the entire curve defined in Eq. (\ref{orbit}).

Since a solution to equation (\ref{EM}) depends on six independent constants: the initial position and initial velocity, an orbit depends on five independent constants. Therefore, one expects that an orbit in a MICZ-Kepler problem with magnetic charge $\mu$ is uniquely determined by the five independent constants ${\mb L}$ and ${\mb A}$. Indeed, that is true  provided that we stick to oriented MICZ-Kepler orbits, a major fact we shall establish in the theorem below.

\smallskip
To state this theorem, we let ${\ms O}$ be the set of oriented MICZ-Kepler orbits, $\ms M$ be the subspace of $\bb R^6$ consisting of pairs $(\mb A, \mb L)$ of vectors in ${\bb R}^3$ such that $L^2 > ({\mb L}\cdot {\mb A})^2$, i.e,
\begin{eqnarray}
{\ms M} =\left\{(\mb A, \mb L)\in {\bb R}^3\times {\bb R}^3 \mid L^2 > ({\mb L}\cdot {\mb A})^2\right\}.
\end{eqnarray}
Note that, for $(\mb A, \mb L)\in \ms M$, $\mb L \neq (\mb L\cdot \mb A) \mb A$, otherwise, ${\mb L}\cdot {\mb L}=(\mb L\cdot \mb A)^2$, a contradiction.

The knowledge about the MICZ Kepler orbits up until now can be summarized in the following theorem, cf. Ref. \cite{MICZ}. 
\begin{Thm}\label{oldfacts} The following statements are true.

(1) For each  $(\mb A, \mb L)\in \ms M$,  the curve defined in Eq. (\ref{orbit}) with $\mu ={\mb L}\cdot {\mb A}$ is a conic with eccentricity
\begin{eqnarray}
e = {|{\mb L}\times {\mb A}|\over |{\mb L}-\mu {\mb A} |}.
\end{eqnarray}
Moreover, 
\begin{eqnarray}
1-e^2={L^2-\mu^2\over |{\mb L}-\mu {\mb A}|^2}(1-A^2),
\end{eqnarray}
and the conic is oriented by ${\mb L}-\mu {\mb A}$ in the sense that ${\mb L}-\mu {\mb A}$ is a positive multiple of the binormal vector of the oriented curve.

(2) The assignment of an oriented conic to each $(\mb A, \mb L)\in \ms M$ in part (1) is one-to-one. So an oriented MICZ-Kepler orbit is uniquely determined by its angular momentum and Lenz vector.

(3) Every conic in part (1) is a MICZ-Kepler orbit. So there is an one-to-one correspondence 
\begin{eqnarray}
\phi: \quad {\ms M}\to {\ms O}
\end{eqnarray}
via equation (\ref{orbit}).

(4) For the oriented MICZ-Kepler orbit  with angular momentum $\mb L$  and Lenz vector $\mb A$, we have the energy formula 
\begin{eqnarray}
E=-{1- A^2\over 2(L^2 -\mu^2)}
\end{eqnarray} where $\mu =\mb L\cdot \mb A$.

(5) The MICZ-Kepler orbit is respectively an ellipse, parabola and hyperbola if $E$ is  respectively negative, zero and positive.

\end{Thm}
The correspondence $\phi$ in part (3) gives the {\bf old parametrization} for the oriented MICZ-Kepler orbits from which we know that the moduli space of all oriented MICZ-Kepler orbits naturally becomes a smooth $6$-manifold. Although this theorem is essentially known to the authors in Ref. \cite{MICZ}, for completeness, we give a proof here.

\begin{proof}
Throughout this proof we shall work under the assumption that $\mu\neq 0$ . The result for $\mu =0$ is then obtained by taking the limit $\mu \to 0$. With this in mind, the curve described in Eq. (\ref{orbit}) can be recast as 
\begin{eqnarray}\label{orbit'}
({\mb L} -\mu {\mb A})\cdot {\mb r} = \mu(L^2-\mu^2), \quad {\mb L}\cdot {\mb r} = \mu r.
\end{eqnarray}
Let $\tilde{\mb L} = {\mb L} -\mu {\mb A}$ and write ${{\mb L}\over \mu}= {{\mb L}\cdot \tilde{\mb L}\over \mu {\tilde L}^2}\tilde{\mb L}+\tilde {\mb A}$. Note that $\tilde{\mb L} \neq \mb 0$, otherwise, $0 = \tilde{\mb L}\cdot {\mb L}=L^2-\mu^2$, a contradiction. Therefore, the curve can be described by equations
\begin{eqnarray}
\tilde{\mb L}\cdot {\mb r} =  \mu(L^2-\mu^2), \quad  r-\tilde{\mb A}\cdot {\mb r} = {(L^2-\mu^2)^2\over \tilde L^2},
\end{eqnarray}
supplemented by the condition $\tilde{\mb L}\cdot \tilde{\mb A} =0$. It is clear that the curve is a conic with eccentricity
$$e=| \tilde{\mb A}|=\sqrt{{L^2\over \mu^2}-{({\mb L}\cdot \tilde{\mb L})^2\over \mu^2 {\tilde L}^2}}=\sqrt{L^2{\tilde L}^2-({\mb L}\cdot \tilde{\mb L})^2\over \mu^2{\tilde L}^2}={|{\mb L}\times \tilde{\mb L}| \over |\mu| {\tilde L}}={|{\mb L}\times {\mb A}| \over | {\mb L} -\mu {\mb A}|}.$$
So $$1-e^2=1- {L^2\over \mu^2}+{({\mb L}\cdot \tilde{\mb L})^2\over \mu^2 {\tilde L}^2}= {L^2-\mu^2\over \mu^2{\tilde L}^2}(L^2-\mu^2-{\tilde L}^2)= {L^2-\mu^2\over {\tilde L}^2}(1-A^2).$$
This is part (1). 

\smallskip
To prove part (2), we make a small digression first. Observe that the projection of the Minkowski space onto ${\bb R}^3$ defines a diffeomorphism from the future light cone
$\{(x_0, {\mb r}) \mid x_0^2 - r^2 =0, x_0>0\}$ onto ${\bb R}^3_*$. Therefore, each oriented curve in ${\bb R}^3_*$ can be lifted to a unique oriented curve on the future light cone. In fact, the oriented curve defined by Eq. (\ref{orbit}) is lifted to the intersection of the future light cone with the oriented plane
defined by the following equations
\begin{eqnarray}\label{newF}
a\cdot x =1, \quad l \cdot x=0.
\end{eqnarray}
Here,  $x = (x_0, {\mb r})$,
\begin{eqnarray}
a={1\over L^2-\mu^2}(1, {\mb A}), \quad l={1\over \sqrt{L^2-\mu^2}}(\mu, {\bf L}),
\end{eqnarray}
and $\cdot$ is the Lorentz product, and the orientation of the plane is represented \footnote{Up to a positive multiple, $*_4(a\wedge l)=*_3({\mb L}-\mu {\mb A})+e_0\wedge *_3({\mb A}\wedge {\mb L}) $, so, once pulled back to the orbit plane on ${\bb R}^3$, it becomes $*_3({\mb L}-\mu {\mb A})$  which represents the orientation of the orbit plane on ${\bb R}^3$. Here, $*_4$ ($*_3$ reps.) is the Hodge star operator of ${\bb R}^4$ ($\bb R^3$ resp.). }  by $a\wedge l$.  Note that $a\cdot l =0$, $l\cdot l = -1$ and $a_0>0$. With this digression in mind, all we need to do is to show that the assignment of the oriented plane defined by Eq. (\ref{newF}) is one-to-one. But this can be verified as follows. 

Suppose that $(\tilde {\mb A}, \tilde {\mb L})$ and $( {\mb A}, {\mb L})$ define the same oriented plane, then
$a\wedge l$ and $\tilde a\wedge \tilde l$ must be a positive multiple of each other. So $\tilde a = x_1a+y_1l$ and $\tilde l = x_2a+y_2l$ for some real numbers $x_1$, $y_1$, $x_2$ and $y_2$. Then
$1= \tilde a \cdot x = x_1a\cdot x+y_1 l\cdot x= x_1$ and $0 =\tilde l \cdot x = x_2 a\cdot x+y_2 l\cdot x = x_2$. So $\tilde l = y_2l$ and $\tilde a = a+y_1l$. Then
$-1=\tilde l^2 = y_2^2 l^2 = -y_2^2$ and $0=\tilde l\cdot \tilde a =-y_1y_2$, so $y_1=0$ and $y_2=1$ or $-1$, i.e., $\tilde a =a$ and $\tilde l =l$ or $-l$. Since $a\wedge l$ and $\tilde a\wedge \tilde l$ must be a positive multiple of each other, we have $\tilde l = l$. This finishes the proof of part (2) because $(a, l)$ and $(\mb A, \mb L)$ determine each other.

\smallskip
To prove part (3), for each $({\mb A}, {\mb L})\in {\ms M}$, we need to specify initial data $\mb q$, $\mb v$ such that
\begin{eqnarray}
{\mb L} = {\mb q}\times {\mb v}+\mu{{\mb q}\over q}, \quad {\mb A} = {\mb L}\times {\mb v}+{{\mb q}\over q}
\end{eqnarray}
where $\mu ={\mb L}\cdot {\mb A} $. To find such initial data, we fix a 2D vector subspace of $\bb R^3$ that contains both $\mb L$ and $\mb A$, then we choose a unit vector $\mb n$ on this subspace such that
$\mu = {\mb L}\cdot {\mb n}$ and $|{\mb A}-{\mb n}|$ is largest possible. Such an $\mb n$ is unique unless ${\mb L}\times {\mb A}=0$. In any case, ${\mb A} \neq {\mb n}$. Now if we
choose
\begin{eqnarray}
{\mb v}={1\over L^2}({\mb A}-{\mb n})\times {\mb L},
\end{eqnarray}
then $\mb v\neq \mb 0$ and is orthogonal to the subspace.  Next we chose 
 $$
 {\mb q} ={{\mb v}\times ({\mb L}-\mu{\mb n})\over v^2}.
 $$
 Note that ${\mb q}\neq \mb 0$, otherwise, $L^2=\mu^2$, contradiction. One can easily see that ${\mb n}\times {\mb q}=0$, so ${\mb q} =\pm q\mb n$, in fact, ${\mb q}=q{\mb n}$. To see that, we note that
 $v^2 {\mb n}\cdot {\mb q}={\mb v}\cdot ({\mb L}\times {\mb n})$, so
 \begin{eqnarray}
 v^2L^2 {\mb n}\cdot {\mb q}&= &(({\mb A}-{\mb n})\times {\mb L})\cdot ({\mb L}\times {\mb n}) = (({\mb A}-{\mb n})\times {\mb L})\times {\mb L})\cdot {\mb n}\cr
 &=& L^2(1-{\mb A}\cdot{\mb n})
 \end{eqnarray}
or
\begin{eqnarray}
 v^2 {\mb n}\cdot {\mb q} &=& 1-{\mb A}\cdot{\mb n}= 1- {\mu^2\over L^2} +\sqrt{1-{\mu^2\over L^2}}\sqrt{A^2-{\mu^2\over L^2}} >0.
 \end{eqnarray}
 One can verify easily that this choice of initial data works:
\begin{eqnarray}
{\mb L}\times {\mb v}+{\mb q\over q} ={1\over L^2}{\mb L}\times((\mb A-\mb n)\times \mb L)+\mb n = (\mb A-\mb n)+\mb n= \mb A.\nonumber
\end{eqnarray}
\begin{eqnarray}
{\mb q}\times {\mb v}+\mu{\mb q\over q} ={1\over v^2}({\mb v}\times(\mb L-\mu\mb n))\times \mb v+\mu\mb n = (\mb L-\mu\mb n)+\mu\mb n= \mb L.\nonumber
\end{eqnarray}

 \smallskip
To prove part (4), we first note that $E={1\over 2}{\mb r'}^2-{1\over r}+{\mu^2\over 2r^2}$.  Next,
\begin{eqnarray}
 A^2 & = & ({\mb L}\times {\mb r}'+{\mb r\over r})^2 = L^2{\mb r'}^2-({\mb L} \cdot {\mb r}')^2+1+{2\over r}{\mb L}\cdot ({\mb r}'\times {\mb r})\cr
 &=& L^2{\mb r'}^2-\mu^2r'^2+1+{2\over r}{\mb L}\cdot (-{\mb L}+\mu {\mb r\over r})\cr
  &=& L^2({\mb r'}^2-{2\over r})+\mu^2({({\mb r}\times {\mb r}')^2\over r^2}-{\mb r}'^2)+1+{2\mu^2\over r}\cr
  &=& (L^2-\mu^2)({\mb r'}^2-{2\over r})+\mu^2{({\mb L}-\mu{\mb r\over r})^2\over r^2}+1\cr
  &=& (L^2-\mu^2)({\mb r'}^2-{2\over r})+\mu^2{L^2-\mu^2\over r^2}+1\cr
   &=&2 (L^2-\mu^2)E+1.\nonumber
 \end{eqnarray}
So
$$
E =- {1-A^2\over 2(L^2-\mu^2)}.
$$

\smallskip
Part (5) is a direct consequence of part (1) and part (4).
\end{proof}
From part (1) of this theorem, it is easy to see that the orbit is a circle if and only if $\mb L \times \mb A =0$.

\section{The new parametrization of orbits}\label{2ndP}
We wish to give a description for the set of the oriented MICZ orbits from the Minkowski space perspective. To do that, we let 
\begin{eqnarray}
{\mathcal M} = \left\{ (a, l) \in {\bb R}^4\times {\bb R}^4\mid l^2 = -1, a\cdot l =0, a_0>0\right\}.
\end{eqnarray}
One can see that the map 
\begin{eqnarray}
\begin{matrix}
{\mathcal M} &\longrightarrow & {\ms M} \cr
(a, l) &\mapsto & ({1\over a_0} {\mb a}, {1\over\sqrt{ a_0}} {\mb l}).
\end{matrix}
\end{eqnarray}
is an one-to-one correspondence. In fact, the inverse map sends $(\mb A, \mb L)$ to 
$$\left({1\over L^2-({\bf L}\cdot {\bf A})^2}(1, {\mb A}), {1\over \sqrt{L^2-({\bf L}\cdot {\bf A})^2}}({\bf L}\cdot {\bf A}, {\bf L})\right).$$
In view of part (3) of Theorem \ref{oldfacts}, we have an one-to-one correspondence
\begin{eqnarray}
\varphi:\quad \mathcal M\to \ms O.
\end{eqnarray}
And this is the {\bf new parametrization} for the oriented MICZ-Kepler orbits. Note that, the investigation in the previous section says that, for $(a, l)\in \mathcal M$, $\varphi (a, l)$
is the intersection of the future light cone with the plane
\begin{eqnarray}
a\circ x =1, \quad l \circ x=0
\end{eqnarray}
oriented by $a\wedge l$. 

\begin{Thm}\label{main} Let $\ms O_+$ ($\ms O_0$ reps.) be the the set of oriented elliptic (parabolic reps.) MICZ-Kepler orbits, ${\mathcal M}_+ =\{(a, l)\in {\mathcal M}\mid a^2 > 0\}$,  and ${\mathcal M}_0 =\{(a, l)\in {\mathcal M}\mid a^2 = 0\}$.

(1) For the oriented MICZ-Kepler orbit  parametrized by $(a, l)\in \mathcal M$, we have the energy formula 
\begin{eqnarray}
E=-{a^2\over 2a_0}.
\end{eqnarray}

(2) $\varphi(\mathcal M_+)={\ms O}_+$, $\varphi(\mathcal M_0)={\ms O}_0$.

(3) The action of ${\mr {SO}}^+(1,3)\times {\bb R}_+$ on $\mathcal M_0$ defined by $(\Lambda, \lambda)\cdot (a, l) = (\lambda\cdot (\Lambda a), \Lambda l)$ is transitive. So ${\mr {SO}}^+(1,3)\times {\bb R}_+$ acts transitively on the set of oriented parabolic MICZ-Kepler orbits.

(4) The action of ${\mr {SO}}^+(1,3)\times {\bb R}_+$ on $\mathcal M_+$ defined by $(\Lambda, \lambda) \cdot (a, l) = (\lambda\cdot \Lambda a, \Lambda l)$ is transitive. So ${\mr {SO}}^+(1,3)\times {\bb R}_+$ acts transitively on the set of oriented elliptic MICZ-Kepler orbits.  

(5) The action in parts (3) and (4) extends to ${\mr {O}}^+(1,3)\times {\bb R}_+$.

\end{Thm}
\begin{proof} Parts (1)and (2) are trivial consequences of Theorem \ref{oldfacts}. For the remaining parts, we first observe that the action of ${\mr {SO}}^+(1,3)$ on ${\bb R}^4\times {\bb R}^4$ defined by $\Lambda \cdot (a, l)=  (\Lambda a, \Lambda l)$ preserves $l^2$, $a\cdot l$, $a^2$, and  the sign of $a_0$ when and only when $a^2\ge 0$. Therefore, this action leaves invariant both $\mathcal M_0$ and $\mathcal M_+$.

To prove that the action of ${\mr {SO}}^+(1,3)\times {\bb R}_+$ on $\mathcal M_0$ is transitive, pick any $(a, l)\in \mathcal M_0$. After a Lorentz boost along the direction of $\mb l$ followed by a space rotation, we may assume that $l=(0, 1, 0, 0)$ and $a=a_0(1, 0, 1, 0)$ with $a_0>0$. After a scaling or a Lorentz boost along $(0,0,1,0)$, $a$ can be further turned into $(1, 0, 1, 0)$. Therefore,  the action of ${\mr {SO}}^+(1,3)\times {\bb R}_+$ on $\mathcal M_0$ is transitive.

To prove that the action of ${\mr {SO}}^+(1,3)\times {\bb R}_+$ on $\mathcal M_+$ is transitive, pick any $(a, l)\in \mathcal M_+$. Since $a^2>0$, after a Lorentz boost, we may assume that
$a=(a_0,{\mb 0})$ with $a_0>0$ and $l=(0, {\mb l})$ with $\mb l \cdot \mb l =1$. After a scaling and a space rotation, we can turn $a$ into $(1, {\mb 0})$, $l$ into $(0, 1, 0, 0)$. Therefore,  the action of ${\mr {SO}}^+(1,3)\times {\bb R}_+$ on $\mathcal M_0$ is transitive.

Part (5) is obvious because ${\mr O}(3)$ acts on $\mathcal M$.
\end{proof}
In the parabolic case, the subgroup group ${\mr {SO}}^+(1, 3)$ already acts transitively. In the
elliptic case, modulo a Lorentz transformation, $(a, l)$ can be assumed to be of this form: $a =(a_0, 0, 0, 0)$ and $l =(0, 1, 0, 0)$, i.e., $A = 0$, $L = {1\over \sqrt{a_0}}$ so that the magnetic charge is zero and the orbit is a circle.  In the hyperbolic case, we can still get a transitive action if we add the oriented ``hyperbolic orbit companions" and their certain limits to the set of oriented parabolic orbits. Here, by a hyperbolic orbit companion we mean the hyperbolic branch that has the longest distance from $\mb 0$, i.e., the focus where the center of mass is located. In terms of the 4D language,  parabolic orbits correspond to $a_0 > 0$, their companions correspond to $a_0 < 0$, and the limits correspond to $a_0 = 0$. 

The Jordan algebra approach to the Kepler problem requires a second temporal dimension, and the Lorentz transformations in this note refer to the one that mixes space with this second temporal dimension. While the Minkowski 4-vector $(1, \mb A)$, the Lenz vector in the Jordan algebra approach to the Kepler problem, can be time-like, light-like and space-like, the Minkowski 4-vector $(\mu, \mb L)$ is always space-like, so the magnetic charge $\mu$ is relative and becomes zero in a certain ``inertial" frame.  Is this second temporal dimension more than just a mathematical artifact?  We wish to know the answer to this question.

\end{document}